\documentclass[11pt]{article}

\usepackage[numbers, compress]{natbib}
\usepackage{caption}
\usepackage{graphicx}

\usepackage{wrapfig}
\usepackage{geometry}
\usepackage{hyperref}
\hypersetup{
    colorlinks=true,
    linkcolor=blue,
    filecolor=magenta,      
    urlcolor=cyan,
    pdftitle={Overleaf Example},
    pdfpagemode=FullScreen,
}

\usepackage[utf8]{inputenc} 
\usepackage[T1]{fontenc}    
\usepackage{url}            
\usepackage{booktabs}       
\usepackage{microtype}      
\usepackage{xcolor}         

\usepackage{authblk}

\usepackage{config}
\usepackage{math_macros}

\title{Computing Nash Equilibria in Potential Games with Private Uncoupled Constraints}
\author[1,2]{Nikolas Patris}
\author[1]{Stelios Stavroulakis}
\author[1,2]{Fivos Kalogiannis}
\author[1]{Rose Zhang}
\author[1]{Ioannis Panageas}
\affil[1]{University of California, Irvine}
\affil[2]{Archimedes Research Unit}
\affil[ ]{\texttt{\{npatris, sstavrou, fkalogia, rwzhang\}@uci.edu,
    ipanagea@ics.uci.edu}}
\date{} 

\begin{document}

\maketitle

\begin{abstract}
We consider the problem of computing Nash equilibria in potential games where each player's strategy set is subject to private uncoupled constraints. This scenario is frequently encountered in real-world applications like road network congestion games where individual drivers adhere to personal budget and fuel limitations. Despite the plethora of algorithms that efficiently compute Nash equilibria (NE) in potential games, the domain of constrained potential games remains largely unexplored. We introduce an algorithm that leverages the Lagrangian formulation of NE. The algorithm is implemented independently by each player and runs in polynomial time with respect to the approximation error, the sum of the size of the action-spaces, and the game's inherent parameters.
\end{abstract}

\section{Introduction}

The modeling and studying of games with constraints has received a lot of attention in various areas including control \cite{decentralized_flow}, transportation and routing in traffic networks \cite{transportation, side_constraints}, telecommunications \cite{telecom}, \cite{telecom2}, markets \cite{arrow1954existence}, cloud computing \cite{cloud_computing}, multi-agent RL \cite{finding_corrrelated}, even in managing environmental pollution \cite{environmental_pollution}. There are two main models in the literature:  settings in which each agent has their own private constraints (also called orthogonal constraints) and settings in which the constraints are the same for all agents (common couple constraints). It is worth noting that having arbitrary constraints might lead to non-existence of a Nash equilibrium, i.e., John Nash's theorem \cite{john_nash} is not applicable, even for two player zero-sum games \cite{altman_constrained_games}. A Nash equilibrium in constrained games, also known as \textit{Generalized Nash Equilibrium}, is defined to be a feasible strategy profile so that each agent does not have an incentive to unilaterally, \textit{feasibly} deviate and decrease their cost. In case of $\epsilon$-approximate Nash equilibrium, we borrow the definition used in \cite{jordan2023first}. Specifically, an $\epsilon$-approximate Nash equilibrium is a strategy profile that is approximately feasible (inequality constraints are violated by at most an additive $\epsilon$) and moreover each agent does not have an incentive to unilaterally deviate among the approximately feasible strategies of their constraints and decrease their cost. 

A common sufficient condition that guarantees existence of Nash equilibria in constrained games is Slater's condition \cite{bertsekas1997nonlinear}. Slater's condition is a constraint qualification (CQ)  stating that for every strategy profile $\vx := (\vx_i,\vx_{-i})$ of the agents (feasible or infeasible), each agent $i$ can deviate and create a strategy profile $(\vx_i',\vx_{-i})$ that is strictly feasible (see \cref{assumption:slater}) for their own constraints. 

In this paper, our primary focus is on the computation of Nash equilibria in potential games, as defined in \cref{section:prelim}, with private constraints, specifically in the context of normal-form games. The constraints we consider in this study are assumed to be convex. The necessity of this assumption to establish arguments regarding approximate Nash equilibria is thoroughly examined in \cref{section:concave_constraints}. While the constraints themselves are assumed to be convex, it is important to note that the potential is non-convex. This non-convexity poses a significant technical challenge that has yet to be addressed in the existing literature. However, in this paper, we  tackle this challenge effectively.
\paragraph{Our contributions} Our paper focuses on the problem of computing approximate Nash equilibrium in normal form potential games with private convex constraints. We present Algorithm \igd that uses projected gradient descent on a carefully constructed Lagrange function that reaches an $\epsilon$-approximate Nash equilibrium after $O(1/\epsilon^6)$ iterations. We note that our algorithm is implemented in a distributed manner in the sense that each agent performs projected gradient on their own cost function and uses information about their own constraints.   

\paragraph{Technical overview}
We present a concise roadmap of the main contributions. In \cref{section:main_result}, we outline the steps taken to prove \cref{thm:main}. Following the standard procedure for constraint optimization, we formulate a Lagrangian problem that incorporates the constraints. Our objective is to identify a stationary point of the Lagrangian. However, a challenge arises when dealing with multipliers $\vlambda$ in bounded domains, as the commonly used first-order methods are designed for unbounded domains. Consequently, there is uncertainty regarding whether the first-order stationary points obtained using these methods are true stationary points in the unconstrained case or if they arise due to the imposed restrictions. To address this issue, we introduce a regularization term to restrict the domain of the lambdas effectively. This regularization introduces a new technical consideration: the choice of regularization parameter significantly impacts the solution quality. Nevertheless, we successfully establish bounds and determine the appropriate relationship between the solution quality and the regularization constant.

\paragraph{Related work} A common technique in constrained optimization is to incorporate the constraints, i.e. define the Lagrange function and optimize over the domain of interest and the Lagrange multipliers (e.g., aim at computing an approximate Karush-Kuhn-Tucker (KKT) point). The KKT condition provides a general characterization of local optimality under various CQs \cite{kuhn_tucker, karush}. Both the KKT condition and CQs can be extended to get a Generalized Nash equilibria. 

The purpose of Lagrange multipliers is to penalize agents in situations where they deviate significantly from feasibility. Several studies have adopted a similar principle known as penalty-type algorithms. However, the selection of the penalty function may vary based on the specific constraints' structure. Notable representatives include exact and inexact penalty methods \cite{Facchinei2011, Fukushima2011, Facchinei_and_Kanzow_2010b,Kanzow_and_Steck_2018, ba_and_pang_2022} and exact and inexact augmented Lagrangian method \cite{Pang2005, Kanzow2016OnTM, Kanzow_and_Steck_2016, Kanzow_and_Steck_2018}. For more information, see \cite{Fisher_et_al_2014} and references therein.
Finally there have been a lot of works based on the Nikaido-Isoda function; these include gradient projection \cite{Rosen_1965},  relaxation method \cite{Uryas, Krawczyk2000, vonHeusinger2009} and Newton methods \cite{Facchinei2009,vonHeusinger2012,dreves2013globalized, izmailov2014error, fischer2016globally}. In the context of potential games the works of \cite{zhu2008lagrangian, 1660950} are regarded as standard. Notably, this extensive body of literature has focused primarily on on convex functions with common, coupled, or uncoupled constraints.
This assumption, however, does not align with real-world problems. Conversely, when the potential function is non-convex, the understanding of the constraint setting is limited. For coupled constraints, there are no guarantees beyond the simplest scenario where constraints are linear, see \cite{sagratella2017algorithms, jordan2023first}. The recent, concurrent, work of \cite{alatur2023provably} considers the more general problem of learning Nash policies in constrained Markov potential games. However, it fundamentally diverges from our approach. In contrast to our decentralized framework, enabling each agent to act independently, the solution proposed in \cite{alatur2023provably} is centralized. This implies that a central authority is responsible for determining which agent (only one) performs a best response update in each iteration.
\section{Preliminaries} \label{section:prelim} 

\subsection{Notation and Definitions}

\paragraph{Notation} Let $\mathbb{R}$ be the set of real numbers, and $[n] = \{ 1, 2, \ldots, n \}$. We define $\Delta$ as the probability simplex, which is the set of $n$-dimensional probability vectors, i.e., $\Delta \coloneqq \{ \vx \in \mathbb{R}^n : x_i \geq 0, \sum_{i=1}^{n} x_i = 1 \}$. We use $\ve_i$ to denote the $i$-th elementary basis vector, and to refer to a coordinate of a vector, we use either $x_i$ or $[x]_i$. The superscripts are used to indicate the iterates of an algorithm. Lastly,  it should be noted that all norms used correspond to the standard Euclidean norm, $\norm{\cdot} = \norm{\cdot}_2$.

\paragraph{Normal-form Games} We consider $n$ players, represented by the set $\calN \coloneqq [n]$, with each player $i \in \calN$ having a set of actions denoted by $\mathcal{A}_i$. The joint action profile is represented by $\bm{a} \coloneqq (a_1, a_2, \cdots, a_n) \in \calA$, where $\calA \coloneqq \times_{i \in \calN} \calA_i$ is the product of the action spaces.
Players may also randomize their strategies by selecting a probability distribution over their set of actions. We use $\bm{x}_i(a_i)$ to denote the probability that player $i$ chooses action $a_i \in \mathcal{A}_i$. Since $\bm{x}_i$ is a probability distribution, it must belong to the probability simplex, which we denote as $\Delta(\mathcal{A}_i)$. The set $\Delta(\mathcal{A}_i)$ consists of all probability vectors $\bm{x} \in \mathbb{R}^{|\mathcal{A}_i|}_{\geq 0}$ satisfying $\sum_{a_i \in \mathcal{A}_i} \bm{x}(a_i) = 1$. The product of simplices is denoted by $\Delta^n \coloneqq \times_{i \in \calN } \Delta(\calA_i)$.

We consider the case where the players are trying to minimize their costs. For a given strategy profile $\va \in \calA$, each player $i$ receives a \textit{cost} $c_i(\va)$, where $c_i : \calA \rightarrow \mathbb{R}$. In the case of a randomized strategy $\vx$, i.e., a probability distribution over $\Delta$, we define the notion of \emph{expected cost} as the expected value of the cost function $C_i(\cdot)$ under the distribution $\vx$. Specifically, the expected cost for player $i$ is defined as $\mathbb{E}_{\va \sim \vx}[C_i(\va)] = C_i(\vx)$.

\begin{definition}[Approximate Nash equilibrium] \label{def:approximate_nash_eq}
A joint strategy profile $(\bm{x}_1^{\star}, \bm{x}_2^{\star}, \ldots,$ $\bm{x}_n^{\star}) \in \Delta^n$ is said to be an $\epsilon$ approximate Nash equilibrium if for any player $i \in \calN$ and any possible unilateral deviation $\vx_i' \in \Delta(\mathcal{A}_i)$, the resulting change in player $i$'s expected cost is no more than $\epsilon$.
\begin{equation}
C_i (\vx_i', \bm{x}^{\star}_{-i}) \geq 
C_i(\vx^{\star}) - \epsilon
\end{equation}
\end{definition}

\begin{remark}
We note that the inequality in \cref{def:approximate_nash_eq} should be adjusted based on the objective of the problem, i.e., whether the players aim to maximize their utility or minimize their cost.    
\end{remark}

\paragraph{Potential Games}
A potential game is a type of game that has a single function 
$\Phi(\bm{x}) : \Delta^n \rightarrow \mathbb{R}$, referred to as a potential function, which captures the incentive of all players to modify their strategies. In other words, if a player deviates from their strategy, then the difference in payoffs is determined by a potential function $\Phi$ evaluated at those two strategy profiles. We express this formally as follows:
\begin{definition}[Potential function] \label{def:potential_function}
Consider a joint action profile $\bm{x} = (\bm{x}_1, \bm{x}_2, \ldots$ $, \bm{x}_n)$. For any player $i \in \calN$ and any unilateral deviation $\vx_i'$, the difference in cost resulting from this deviation is reflected in the change in the potential function.
\begin{equation}
C_i(\vx_i', \bm{x}_{-i}) - C_i(\bm{x}) = 
\Phi(\vx_i', \bm{x}_{-i}) - \Phi(\bm{x})
\label{eq:potential_function}
\end{equation}
\end{definition}

\subsection{Problem Statement}
This work addresses the problem of constrained potential games. In these games, multiple players are involved, and each player $i \in \calN$ has a private set of convex constraints. These constraints specify a feasibility set for each player, within which they aim to minimize their individual cost function $C_i(\vx)$.

As noted before, the potential function in constrained potential games captures the incentives of all players simultaneously. This allows us to formulate a single problem, referred to as the primal problem, that captures the individual players' problems collectively. By solving the primal problem, we are able to solve the constituent problems of each player simultaneously.
\begin{definition}[Primal problem]
\label{def:primal_problem}
Let each player $i \in \calN$ have a set of $d_i$ private \textit{convex} constraints, denoted by $g_{i,m}(\cdot)$ for any $m \in [d_i]$. Then, the players aim to solve the following optimization problem.
\begin{equation} \label{eq:primal_problem}
\begin{array}{ccc}
  \text{minimize}   & \Phi(\vx_1, \vx_2, \ldots, \vx_n) & \\
  \text{subject to} & g_{i,m}(\bm{x}_i) \leq 0, & m = 1, 2, \ldots, d_i, \ \forall i \in \calN \\ 
\end{array}
\end{equation}
\end{definition}

The total number of constraints is denoted by $d = \sum_{i \in \calN} d_i$, where $\calN$ represents the set of all players. Additionally, we refer to the feasibility set of player $i$ as $\calS_i$, which is defined by their respective constraints, $g_{i,m}(\bm{x}_i) \leq 0$ for any $m \in [d_i]$.
\begin{equation}
\label{eq:feasibility_sets_definition}
\calS_i = \{ \vx_i \in \Delta(\calA_i) \, | \, g_{i,m}(\vx_i) \leq 0 \text{ for all } m \in [d_i] \}
\end{equation}

As expected, optimization in the constrained case presents additional difficulties. However, this is a well-studied class of optimization problems, and there are standard tools available, with the most prominent being the Lagrangian methods.
The Lagrangian method involves defining a modified problem that takes the constraints into account. Specifically, we begin with the primal problem defined in \cref{def:primal_problem}, and then define a new function that incorporates the constraints.
\begin{definition}[Lagrangian function] 
\label{def:lagrangian_function}
Consider a function $\Phi(\cdot)$ to be minimized, subject to the constraints $g_{i,m}(\vx_i)$ for any player $i \in \calN$ and any constraint $m \in [d_i]$. For each inequality constraint, we introduce a non-negative multiplier $\lambda_{i,m}$, commonly known as a Lagrange multiplier.
\begin{align} 
\calL (\vx, \vlambda) 
&= \Phi(\vx) + \sum\limits_{i=1}^{n} \sum\limits_{m=1}^{d_i} \lambda_{i,m} g_{i,m}(\vx_i) \nonumber\\
&= \Phi(\vx) + \sum\limits_{i=1}^{n} \vlambda_i^{\top} \vg_i(\vx_i) \nonumber\\
&= \Phi(\vx) + \vlambda^{\top} \vg
\label{eq:lagrangian_function}
\end{align}
where $\vlambda$ is a $d$-dimensional vector corresponding where $d \coloneqq \left( \sum_{i \in \calN} d_i \right)$ is the overall number of constraints. The vector $\vlambda$ is the concatenation of $n$ vectors, where each vector has a length corresponding to the number of constraints for the corresponding player, i.e., $d_1, d_2, \ldots, d_n$, respectively. 

\begin{equation}
\vlambda = 
\begin{bmatrix}
\vlambda_1 \\ \vdots \\ \vlambda_n
\end{bmatrix}
\in 
\mathbb{R}^{d_1} \times \ldots \times \mathbb{R}^{d_n} = \mathbb{R}^{d}
\quad
\text{and}
\end{equation}
\begin{equation}
\quad
\vg_i(\vx_i)
=
\begin{bmatrix}
g_{i,1}(\vx_i) \\ \vdots \\ g_{i,d_i}(\vx_i) 
\end{bmatrix}
\in \mathbb{R}^{d_i}
\quad
\text{and}
\quad
\vg 
=
\begin{bmatrix}
\vg_1(\vx_1) \\ \vdots \\ \vg_n(\vx_n) 
\end{bmatrix}
\end{equation}
where $\vlambda_i \in \mathbb{R}_{+}^{d_i}$ for any $i \in \calN$ and $\lambda_{i,m}$ correspond to the Lagrange multiplier of the $m$-th constraint of the $i$-th player.
\end{definition}

\paragraph{Lagrangian dual function} We define the Lagrangian dual function as the minimum value of the Lagrangian \eqref{eq:lagrangian_function} over $\vx \in \Delta$. It's important to note that this minimum can be achieved even if $\vx$ is not a feasible point, meaning it may not satisfy the constraints. 

\begin{definition} [Lagrangian dual function] \label{eq:lagrangian_dual_function} The Lagrangian dual function is defined as $h(\vlambda) = 
\inf\limits_{\vx \in \Delta} 
\calL (\vx, \vlambda) $ where $\vx$ belongs to the product of simplices of all players.
\end{definition}

The Lagrangian dual function provides a lower bound on the optimal value of the primal problem defined in \cref{def:primal_problem}, and this property is known as weak duality. The proof of this inequality is standard and can be found in many textbooks on optimization.

\begin{definition}[Weak duality] \label{def:weak_duality}
Consider the primal problem in \cref{eq:primal_problem} and suppose $p^{\star}$ is its optimal value. Then, for any $\vlambda \succeq \bm{0}$ we have $h(\vlambda) \leq p^{\star}$.
\end{definition}

\paragraph{Dual problem} As we described above, the Lagrangian dual function provides a lower bound on the optimal value $p^{\star}$ of the primal problem for any value of $\vlambda$. It is natural to ask what the best possible lower bound for $p^{\star}$ would be. This leads to the formulation of another optimization problem commonly referred to as the \emph{dual problem}.

\begin{definition}[Dual problem]
\label{def:dual_problem}
Suppose the primal problem as defined in \cref{def:primal_problem}. Then, the dual problem is defined as follows.
\begin{equation} \label{eq:dual_problem}
\begin{array}{ccc}
  \text{maximize}   & h(\vlambda) \textrm{ \;\;}  \text{subject to} & \vlambda \succeq \bm{0}
\end{array}
\end{equation}
\end{definition}

There are also equivalent formulation of the primal, dual problems with respect to each player having the joint strategy profile of the other players $\vx_{-i}$ fixed. We defer those definitions along with some relevant lemmas to the \cref{section:appendix_definition}.
\section{Nonconvex Constrained Games and Lagrangian Functions}

In this section, we provide a concise overview of nonconvex games. We begin by introducing the problem and defining the appropriate solution concept for this class of games. Next, we delve into the importance of having convex constraints and analyze the implications when convexity is violated. Finally, we introduce the concept of a regularized Lagrangian, which plays a crucial role in our proposed solution.

\subsection{Nonconvex Games}
\label{section:nonconvex_games}

It is evident that the goal is to find an Nash equilibrium $\vx$ for the constrained potential game. However, unlike the unconstrained case where the optimization is done over the entire space, in the constrained case, the optimization is restricted to a specific convex domain. In our case, the convex set is the product of the simplices, and so the notion of Nash equilibrium is defined as follows.

\begin{definition}[Approximate First Order Stationary Point] 
\label{def:approximate_stationary_points}
A joint strategy profile $\vx \coloneqq (\vx_1, \cdots, \vx_n) \in \Delta^n$ is called $\epsilon$ approximate first order stationary point of function $f$ as long as 
\begin{equation}
- \min\limits_{
(\vx + \vdelta) \in \Delta^n, 
\norm{\vdelta}^{2} \leq 1
}
\vdelta^{\top} \nabla_{\vx} f(\vx) 
\leq 
\epsilon
\end{equation}
\end{definition}

\paragraph{Solution Concept} The solution concept we consider in this work is commonly referred to as a nonlinear generalized Nash equilibrium \cite{jordan2023first}. Unlike in the case where there are no constraints, in the constrained case, we extend the set to include unilateral deviations that are also approximately feasible. The formal definition of this concept is provided below.

\begin{definition}[Approximate feasible Approximate Nash equilibrium]
\label{def:approximate_feasible_approximate_nash}
A joint strategy profile $(\bm{x}_1^{\star}, \bm{x}_2^{\star}, \ldots,$ $\bm{x}_n^{\star}) \in \Delta^n$ is said to be an $O(\epsilon)$ approximate feasible approximate Nash equilibrium if for any player $i \in \calN$ and any possible unilateral deviation $\vx_i' \in \{ \vx_i \in \Delta(\mathcal{A}_i) \, | \, \vg_i(\vx_i) \leq \epsilon \}$, the resulting change in player $i$'s expected cost is no more than $O(\epsilon)$.
\begin{equation}
C_i (\vx_i', \bm{x}^{\star}_{-i}) \geq 
C_i(\vx^{\star}) - O(\epsilon)
\end{equation}
\end{definition}

\paragraph{Techniques and Assumptions}
To solve our problem, we leverage an important property of the potential function: convexity per player. This property, together with a standard assumption used in the literature (referred to as Slater’s condition), yields a stronger condition (strong duality) for the relationship between the primal and dual optimal values.



\begin{assumption}[Slater's Condition] \label{assumption:slater}
For any player $i \in [n]$ and any constraint $m \in [d_i]$, there is strategy profile $\Tilde{\vx}_i$ such that $g_{i,m}(\Tilde{\vx}_i) < \xi_{i,m}$ for a strictly negative $\xi_{i,m}$.
\end{assumption}

Slater's condition requires the existence of a strictly feasible point under the constraints. Constraints qualifications, such as Slater's condition, are a standard way to obtain strong duality. It is worth noting that Slater's condition only requires the existence of a point that strictly satisfies the constraints, and does not require an a priori knowledge of any candidate optimal solution, such as regularity.

\begin{lemma}[Strong duality per player]
\label{lemma:strong_duality_original}
For any player $i \in [n]$ and for any joint strategy of the other players $\vx_{-i}$ along with their Lagrange multipliers $\vlambda_{-i}$, strong duality holds.

\begin{align}
\min\limits_{\vx_i \in \Delta(\calA_i)} 
\max\limits_{\vlambda_i \succeq \bm{0}} 
\calL_i (\vx_i, \vlambda_i ; \vx_{-i}, \vlambda_{-i}) \nonumber \\
=
\max\limits_{\vlambda_i \succeq \bm{0}}
\min\limits_{\vx_i \in \Delta(\calA_i)} 
\calL_i (\vx_i, \vlambda_i ; \vx_{-i}, \vlambda_{-i})
\label{eq:strong_duality_original}
\end{align}

where $\calL_i$ is the player-wise Lagrangian function defined in Appendix.
\end{lemma}

\begin{remark} \label{remark:1}
In the LHS of \cref{eq:strong_duality_original}, the minimizer $\vx$ selects their strategy first. It is crucial for $\vx$ to choose a feasible point that satisfies the constraints; otherwise, the Lagrange multipliers $\vlambda$ can be set arbitrarily large. Moreover, if $\vx$ is a feasible point, then it implies that the Lagrange multipliers $\vlambda$ must be zero, which results in an alternative expression for strong duality.

\begin{equation} \label{eq:strong_duality}
\min\limits_{\vx_i \in \calS_i} 
\Phi(\vx_i, \vx_{-i})
=
\max\limits_{\vlambda_i \succeq \bm{0}}
\min\limits_{\vx_i \in \Delta(\calA_i)} 
\calL_i (\vx_i, \vlambda_i ; \vx_{-i}, \vlambda_{-i})
\end{equation}

where $\calS_i$ is the feasibility set of player $i$.
\end{remark}

\begin{assumption}
\label{assumption:smooth_constraints}
For any player $i \in \calN$ and any constraint $m \in [d_i]$, the function $g_{i,m}$ is convex and $\gamma$-smooth.
\end{assumption}
\subsection{Necessity of Convex Constraints} \label{section:concave_constraints}

In this subsection, we provide a brief explanation of why the convexity of the constraints is a necessary condition for a first-order methods algorithm to find a Nash equilibrium. 

In \cref{fig:nonconcave_constraints}, we encounter a problematic scenario that highlights a key challenge. Let's assume we have applied a first-order method and found a stationary point $x$. This stationary point suggests that in the neighborhood of $x$, further improvements in the objective function are not possible.

\begin{figure}
    \centering    \includegraphics[width=0.3\linewidth]{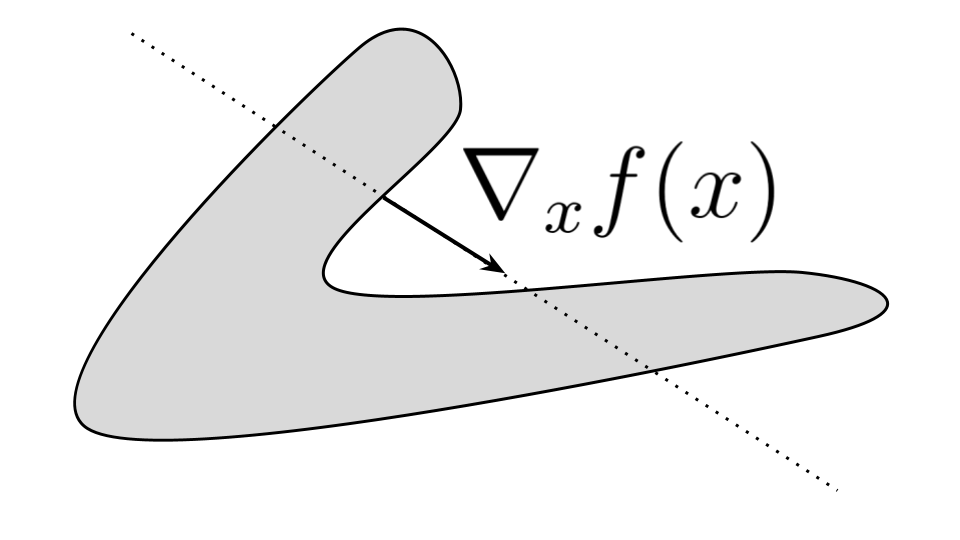}
    \caption{Nonconvex feasibility set}
    \label{fig:nonconcave_constraints}
\end{figure}

This scenario can also be interpreted within the context of zero-sum games. Consider the player $x$, who acts as the minimizer. In this situation, $x$ does not have an incentive to move left as it would lead to being outside the feasible set. Hence, then the maximizer $y$, operating as the Lagrange multiplier in \ref{def:lagrangian_function}, can penalize $x$ by setting $y$ to infinity. Conversely, if $x$ moves right, the objective performance would worsen. This phenomenon captures the notion of a local min-max point.

\begin{definition}[Local min-max point \cite{daskalakis2018limit}]
\label{def:local_min_max_points}
A critical point $(x^{\star}, y^{\star})$ is a local min-max point if there exists a neighborhood $U$ around $(x^{\star}, y^{\star})$ so that for all $(x,y) \in U$ we have that $f(x^{\star}, y) \leq f (x^{\star}, y^{\star}) \leq f (x, y^{\star})$.
\end{definition}

However, relying solely on the concept of a local min-max point is insufficient for our objective. As depicted in Figure \ref{fig:nonconcave_constraints}, although neither player has an incentive to change their strategy within a local neighborhood (i.e., for $\epsilon$ deviations), we observe that there exists a direction along the extension of the gradient line, which, unfortunately, falls into the feasibility set.
\subsection{Regularized Lagrangian}
\label{section:regularized_lagrangian}

In this subsection, we provide a brief explanation of why regularization is an essential component of our proposed solution. To understand this, let us discuss the method of Lagrange multipliers briefly. For a more detailed exposition, we refer to \cite{boyd2004convex}.

\begin{definition}[Regularized Lagrangian function]
Suppose $\calL (\vx, \vlambda)$ is the Lagrangian function as defined in \cref{def:lagrangian_function}. We can define the regularized version $\Tilde{\calL}$ by adding the regularization term $ - \mu \norm{\vlambda}^2$ as follows:
\begin{equation} \label{def:reg_lagrangian_function}
\Tilde{\calL} (\vx, \vlambda) \coloneqq \calL (\vx, \vlambda) - \mu \norm{\vlambda}^2
\end{equation}
\end{definition}

As previously mentioned, the intuition behind the Lagrange method is to introduce a new function that incorporates the constraints, transforming the problem into an unconstrained one. The Lagrange multipliers act as linear penalty terms on the objective. Consequently, whenever $\vx$ is an infeasible point, that means it does not satisfy the constraints, the maximization in the dual problem (as defined in \cref{def:dual_problem}) can set $\vlambda$ arbitrarily large. 

The key distinction lies in the domain of the Lagrange multipliers $\vlambda$. In the original formulation, the domain of $\vlambda$ is unbounded, allowing for arbitrary values. Most widely-used methods for computing first-order stationary points are designed specifically for optimization problems with bounded domains. However, when we introduce bounded domains for the Lagrange multipliers $\vlambda$, a potential issue arises: we cannot determine if the first-order stationary points obtained using these methods are genuine $(\vlambda \succeq \bm{0})$ stationary points in the unconstrained case or if they are introduced due to the restriction. In other words, these bounded domains may prevent the optimizer from reaching the true optimal point that would exist in the unconstrained case.

The regularization term in our formulation plays an important role for two reasons. Firstly, it indirectly limits the magnitude of the Lagrange multipliers. Although one might argue that this modifies the original game, it is important to note that as the regularization parameter $\mu$ approaches zero, the regularized game converges to the original game. Moreover, for sufficiently small values of $\mu$, it governs the alternation of the game. In other words, it allows us to control the impact of the regularization based on how closely the solutions approach (in value) the unconstrained case. As we will demonstrate shortly, the parameter $\mu$ also controls the approximation of the first-order stationary point. Another important aspect of regularization is that it makes the problem with respect to $\vlambda$ strongly concave. Strong concavity guarantees the uniqueness of a maximizer and assists the analysis.
\textit{}\section{Main Result} \label{section:main_result}

The main contribution of this section is divided into two parts. First, we present an algorithm specifically designed for computing first-order stationary points, as defined in \cref{def:approximate_stationary_points}. This algorithm serves as the core computational tool of our approach. Then, we provide a complete set of statements and proofs that establish the correctness of our solution. These statements form the basis for understanding the theoretical underpinnings of our approach.

\subsection{Algorithm \texorpdfstring{\igd}{IGD}} 
\label{section:algorithm}
We present here our proposed solution. Algorithm \igd is a natural and intuitive procedure that essentially performs (projected) gradient descent in a special function, $\phi(\vx) = \max_{\vlambda} \Tilde{\calL} (\vx, \vlambda)$. The first step of the algorithm involves a maximization step, which can be efficiently performed without the need for a max-oracle, thanks to the structure of $\phi(\vx)$ as explained in \cref{section:analysis}. Then, each agent independently performs a step of projected gradient descent on $\Tilde{\calL}$. However, despite its simplicity, Algorithm \igd provides strong guarantees, as we will demonstrate next.
\begin{algorithm}[t]
\caption{ \igd : Independent Gradient Descent on $\phi(\cdot) = \max_{\vlambda} \Tilde{\calL} (\cdot, \vlambda)$ of the regularized Lagrangian $\Tilde{\calL}$}
\label{algo}
\textbf{Output: } $\vxhat$ first order stationary point. \\
\textbf{Initialize: } $\vx^{(0)} \in \Delta^n $\\
\For{$t = 1, 2, \ldots, T$}
{
$
\vlambda^{(t)} 
\leftarrow 
\argmax\limits_{\vlambda}
\left( \vlambda^{\top} \vg(\vx^{(t)}) - \mu \norm{\vlambda}^2 \right)
$
\\
$\vx_{i}^{(t+1)} 
\leftarrow  
\operatorname{\Pi}_i
\left( 
\vx_{i}^{(t)} 
-
\eta 
\left(
\nabla_{\vx_i} C_i(\vx^{(t)}) + \vlambda_{i}^{\top} \nabla_{\vx_i} \vg_i(\vx_{i}^{(t)})
\right)
\right)
$ \algcomment{for all players $i \in \calN$}
}
\textbf{Return} $\vxhat = (\vxhat_1, \ldots, \vxhat_n)$
\end{algorithm}
\subsection{Analysis of Algorithm \texorpdfstring{\igd}{IGD}} 
\label{section:analysis}
In this section, we present a complete proof of our main theorem. We provide a clear roadmap outlining the steps we will take to establish the proof. First, we will argue about the smoothness of the regularized Lagrangian function. By doing so, we will be able to establish the smoothness of the function $\phi(\cdot) = \max_{\vlambda} \Tilde{\calL} (\cdot, \vlambda)$. Next, we will prove some boundedness results on the values of $\vlambda$, both for the case of (approximately) optimal values and for arbitrary values. These bounds are important to ensure that the optimization process remains well-behaved.

\begin{lemma}[Strongly Concavity] \label{lemma:unique_maximizer}
Let $\vxhat$ be an arbitrary joint strategy profile. Then, the function $\Tilde{\calL} (\vxhat, \vlambda) = \calL (\vxhat, \vlambda) - \mu \norm{\vlambda}^{2}$ is strongly concave in $\vlambda$, and so the maximizer $\vlambdahat = \argmax_{\vlambda} \Tilde{\calL} (\vxhat, \vlambda) $ is unique.
\end{lemma}


\begin{lemma}[Bounded norm of multipliers] \label{lemma:bounded_maximizer}
Let $\vxhat$ be an arbitrary joint strategy profile. Then, the maximizer $\vlambdahat = \argmax_{\vlambda} \{ \Tilde{\calL}( \vxhat, \vlambda) \} $ has bounded norm, i.e. $\norm{\vlambdahat} \leq \frac{\sqrt{d} G_{\maxtag}}{2 \mu} = \Lambda_{\maxtag}$, where $G_{\max} = \max\limits_{i \in \calN} \max\limits_{m \in [d_i]} \max\limits_{\vx_i \in \Delta(\calA_i)} g_{i,m}(\vx_i)$.
\end{lemma}

Having established the strong concavity and boundedness of the Lagrange multipliers, we now continue on showing the smoothness of $\Tilde{\calL}$. We use $A_{\maxtag}$ to denote the maximum number of actions over all agents, i.e. $A_{\maxtag} \coloneqq \max\limits_{i \in \calN} |\calA_i|$ and $\Phi_{\max} \coloneqq \max\limits_{\vx} \Phi(\vx)$ for the maximum value of the potential function.

\begin{lemma}[Smoothness of $\Phi$] \label{lemma:smoothness_potential}
The potential function $\Phi$ is $(n A_{\maxtag} \Phi_{\maxtag})$-smooth. 
\end{lemma}

\begin{lemma}[Smoothness of $\calL$] \label{lemma:smoothness_lagrangian}
The Lagrangian function $\calL$ is $(n A_{\maxtag} \Phi_{\maxtag} + \Lambda_{\maxtag} \gamma)$-smooth,
where $\gamma$ is the smoothness of the constraints.
\end{lemma}


By combining \cref{lemma:smoothness_potential} and \cref{lemma:smoothness_lagrangian}, and applying the triangle inequality, we can establish the boundedness of the spectral norm of $\Tilde{\calL}$. 

\begin{lemma}[Smoothness of $\Tilde{\calL}$] 
\label{lemma:smoothness_reg_lagrangian}
The regularized Lagrangian $\Tilde{\calL}$ is $(n A_{\maxtag} \Phi_{\maxtag} + \Lambda_{\maxtag} \gamma+2 \mu)$-smooth.
\end{lemma}

For clarity, let us define $\ell$ as the smoothness parameter of the regularized Lagrangian, such that $\ell = (n A_{\maxtag} \Phi_{\maxtag} + \Lambda_{\maxtag} \gamma + 2 \mu)$ from \cref{lemma:smoothness_reg_lagrangian}. The following proposition follows from the above statements.

\begin{proposition}[Properties of $\Tilde{\calL}$]
\label{prop:properties_reg_lagrangian}
The regularized Lagrangian $\Tilde{\calL}$ satisfies the following properties:

\begin{enumerate}
\item $\Tilde{\calL}$ is $\ell$-smooth and $\Tilde{\calL}(\vx, \cdot)$ is $\mu$-strongly concave.
\item The domain of $\vlambda$ is bounded, i.e. $\norm{\vlambda} \leq \Lambda_{\maxtag}$.
\end{enumerate}
\end{proposition}

Let $\kappa = \ell / \mu$ denote the condition number and define $\phi(\cdot) = \max_{\vlambda} \Tilde{\calL} (\cdot, \vlambda)$ and $\vlambda^{\star} (\cdot) = \argmax \Tilde{\calL} (\cdot, \vlambda)$.

\begin{lemma}[Lemma 4.3 \cite{lin2020gradient}]
\label{lemma:smoothness_max}
Under \cref{prop:properties_reg_lagrangian}, $\phi(\cdot)$ is $(\ell + \kappa \ell)$-smooth with $\nabla \phi(\cdot) = \nabla_{\vx} \Tilde{\calL} (\cdot, \vlambda^{\star}(\cdot))$. Also $\vlambda^{\star}(\cdot)$ is $\kappa$-Lipschitz.
\end{lemma}

\begin{lemma}[Approximate Stationary Point of $\phi$]
\label{thm:stationary_reg_lagrangian}
Let the learning rate $\eta$ be $1/\beta$, where $\beta = c / \mu$ and $c = 4((n A_{\maxtag})^2 + (\Lambda_{\maxtag} \gamma)^2)$ is constant. If we run Algorithm \igd for $T = \frac{32}{\epsilon^2 \mu}
\big( \Phi_{\maxtag} + \Lambda_{\maxtag} \sqrt{d} G_{\maxtag} \big)
\left(
(n A_{\maxtag})^2 + (\Lambda_{\maxtag} \gamma)^2
\right)$, then there exists a timestep $t \in \{ 1, 2, \cdots, T \}$ such that $\vx^{(t)}$ is an $\epsilon$-approximate first order stationary point of $\phi$.
\end{lemma} 


\begin{proof}
The first step involves demonstrating that the function $\phi(\cdot) = \max_{\vlambda} \Tilde{\calL} (\cdot, \vlambda)$ is a smooth function. Based on the result from \cref{lemma:smoothness_reg_lagrangian}, we know that $\Tilde{\calL}$ is $\ell$-smooth, with $\ell = (n A_{\maxtag} \Phi_{\maxtag} + \Lambda_{\maxtag} \gamma + 2 \mu)$. Utilizing the smoothness result in \cref{lemma:smoothness_max}, we can conclude that $\phi(\cdot)$ is also smooth, with a parameter $(\ell + \kappa \ell)$, where $\kappa = \ell / \mu$.

Next, by applying the descent lemma (\cref{lemma:descent_lemma}) on $\phi$, we establish that the sequence of successive steps will invariably form a non-increasing sequence. This can be inferred from \cref{lemma:descent_lemma}, which ensures that $\phi(\vx^{(t+1)}) - \phi(\vx^{(t)}) \leq - \frac{1}{2\beta} \norm{\vx^{(t+1)} - \vx^{(t)}}_{2}^{2}$. Furthermore, \cref{thm:descent_gradient} guarantees the existence of at least one timestep $t \in \{ 1, 2, \cdots, T \}$ where $\norm{\vx^{(t+1)} - \vx^{(t)}}_{2}^{2} \leq \frac{2 \beta \delta_{\phi}}{T}$, with $\delta_{\phi} = \phi(\vx^{(0)}) - \phi(\vx^{\star})$. Bounding $\delta_{\phi}$ by $\big( \Phi_{\maxtag} + \Lambda_{\maxtag} \sqrt{d} G_{\maxtag} \big)$, setting $\beta = c / \mu$ and 
\[
T = \frac{32}{\epsilon^2 \mu}
\big( \Phi_{\maxtag} + \Lambda_{\maxtag} \sqrt{d} G_{\maxtag} \big)
\left(
(n A_{\maxtag})^2 + (\Lambda_{\maxtag} \gamma)^2
\right)
\]




we get that $\norm{\vx^{(t+1)} - \vx^{(t)}}_{2} \leq \epsilon/2$. Then, from \cref{prop:cone_to_nash} and \cref{remark:appendix} we conclude that $\vx^{(t+1)}$ is an $\epsilon$-approximate first order stationary point of $\phi$.
\end{proof}

To establish the time complexity, we are free to choose the value of $\mu$. By setting $\mu = O(\epsilon)$, we determine the dependence of the running time of Algorithm \igd on $\epsilon$.

\begin{corollary}
Algorithm \igd runs in $T = O(1/\epsilon^6)$ iterations to find an $\epsilon$ approximate first order stationary point of $\phi$.
\end{corollary}

To proceed, we can utilize Lemma \ref{lemma:smoothness_max}, which states that $\nabla \phi(\cdot) = \nabla_{\vx} \Tilde{\calL} (\cdot, \vlambda^{\star}(\cdot))$, where $\vlambda^{\star}(\cdot)$ represents the unique maximizer for a given $\vx$, as indicated in Lemma \ref{lemma:unique_maximizer}. Therefore, this implies that the point $\vxhat$ obtained from the algorithm is a first order stationary point of the regularized Lagrangian. Additionally, since the regularization term does not depend on the $\vx$ variable, it follows that $\vxhat$ is a first order stationary point of the Lagrangian as defined in \ref{def:lagrangian_function}.

The next step involves demonstrating that $\vlambdahat = \argmax_{\vlambda} \calL(\vxhat, \vlambda)$ necessarily resides in the interior of its domain, meaning $\norm{\vlambda} \leq \Lambda_{\maxtag}$. More precisely, we are looking for a bound that is independent of $\mu = O(\epsilon)$. To accomplish this, we utilize the Slater's condition.

\begin{lemma}[Bounded optimal multipliers]
\label{lemma:bounded_optimal_multipliers}
Let $\vxhat$ represent the $\epsilon$-approximate first order stationary point returned by Algorithm \igd. Then, for each set of Lagrange multipliers $\vlambdahat_i$, where $\vlambdahat_i = \argmax_{\vlambda_i} \calL_i(\vxhat_i, \vlambda_i ; \vxhat_{-i}, \vlambdahat_{-i})$, we can bound them by $\frac{2 \big( \Phi_{\maxtag} - \Phi_{\mintag} \big) }{\xi_{i,m}}$ component-wise, where $\xi_{i,m}$ are defined in \cref{assumption:slater}.
\end{lemma}

Finally, we need to consider the (approximate) optimality with respect to the potential function $\Phi$. So far, our discussion has mainly focused on the Lagrangian, but we have not directly addressed the optimality of the potential function and as a result the individual players' value functions. In other words, we need a relation that associates the first order stationary point with a guarantee with respect to the value of $\Phi$.

\begin{lemma}[Approximate Optimality of $\Phi$]
\label{lemma:approximate_optimality}
Given an $\epsilon$ approximate stationary point $(\vxhat, \vlambdahat)$ of $\Tilde{\calL}$, we get that $\vxhat$ is $O(\epsilon)$ approximate feasible approximate Nash equilibrium, as in \cref{def:approximate_feasible_approximate_nash}.
\end{lemma}

\begin{theorem}[Main Theorem] 
\label{thm:main}
Assuming that all agents perform Algorithm \igd, after $T = O(1/\epsilon^6)$ steps, there exists an iterate $\vx^{(t)}$ for $t \in [T]$, so that $\vx^{(t)}$ is an $O(\epsilon)$ approximate feasible approximate Nash equilibrium.
\end{theorem}

\section{Numerical Experiments} 
\label{section:experiments}

\begin{figure}[t]
    \centering
    \includegraphics[width=0.5\linewidth]{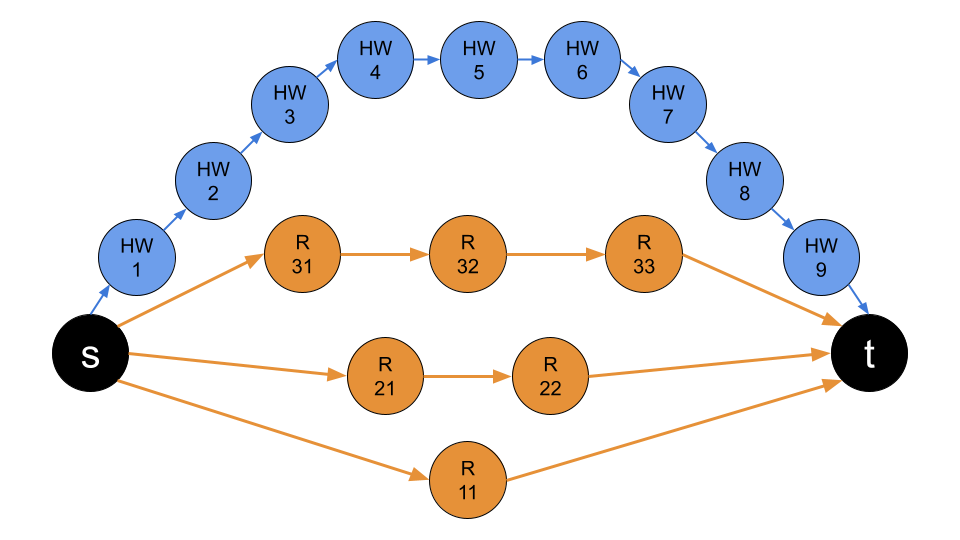}
    \caption{Players can pick out of four options: \texttt{R1, R2, R3, HW} by consuming \texttt{\{2,3,4,10\}} from their gas allowance.}
    \label{fig:network}
\end{figure}

In this section, we empirically validate our theoretical results using constrained congestion games as a testing ground. Relying on the framework of congestion games is possible since every congestion game can be interpreted as a potential game \cite{Rosenthal1973}.


\paragraph{Experimental Setup}
Our experimental setup involves a rooted directed acyclic graph (DAG), as illustrated in \cref{fig:network}. The graph comprises four paths connecting a source node $s$ and a target node $t$. The length of each path is determined by the number of edges it contains, given that every edge in the graph is of unit distance. In addition, there are five players, each assigned the task of selecting a path and subject to an additional constraint on the amount of gas they are allowed to expend. The gas expenditure is proportional to the distance traveled, with one unit of gas allowing the traversal of one unit of distance. The congestion experienced on each path is influenced by the number of players selecting that particular route. Although the congestion functions for all paths are linear, with the yellow paths sharing the same function, the highway's function (blue path) has a smaller parameter, resulting in a lower cost experienced. This trade-off between congestion and gas consumption on the highway presents the important decision-making factor for the players.\footnote{The code is available at the GitHub repository: https://github.com/steliostavroulakis/constrained-potential-games}

\paragraph{Implementation \& Results}
\begin{wrapfigure}[]{r}{0.5\textwidth}
    \includegraphics[width=0.8\linewidth]{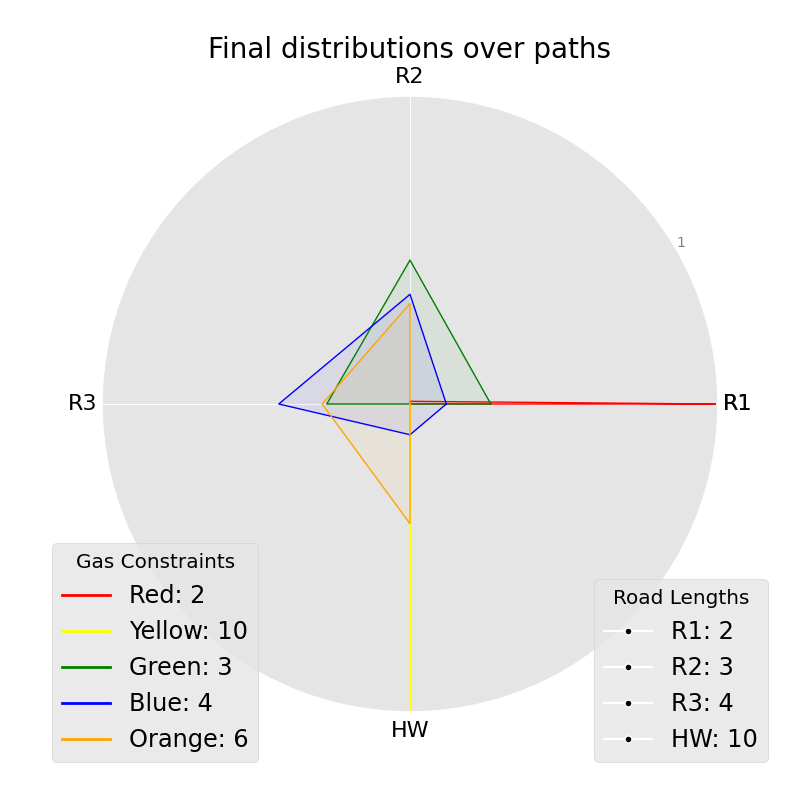}
    \caption{Spider chart capturing the final distributions over paths of all players.}
    \label{fig:spider}
\end{wrapfigure}
Upon completing the iterations of Algorithm \igd, we generate a spider chart to visually represent the final distributions of the players. Additionally, a table detailing the gas constraints for each player is presented alongside the spider chart, see \cref{fig:spider}. The agents select paths with lengths proportional to their gas constraints while simultaneously minimizing congestion.

\begin{figure}[t]
    \centering
    \begin{subfigure}[t]{0.3\textwidth} 
        \includegraphics[width=\textwidth]{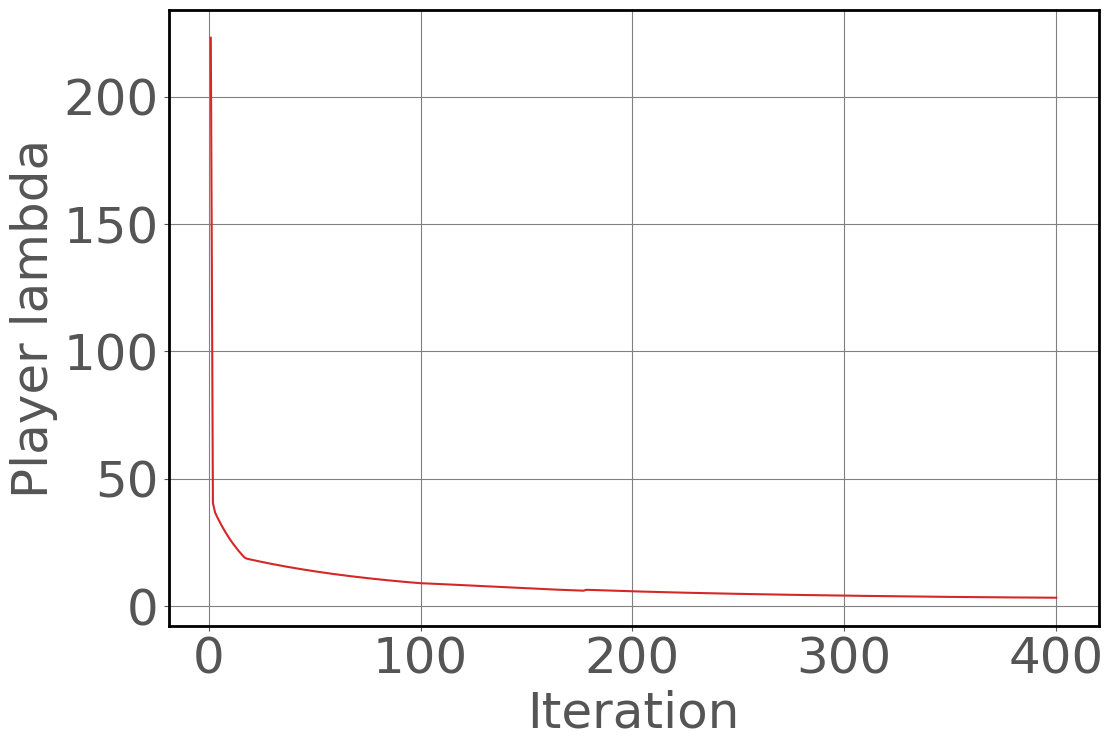}
        \caption{Sum of Lagrangian multipliers}
        \label{fig:lagrangian_mult}
    \end{subfigure}
    \hfill
    \begin{subfigure}[t]{0.3\textwidth}
        \includegraphics[width=\textwidth]{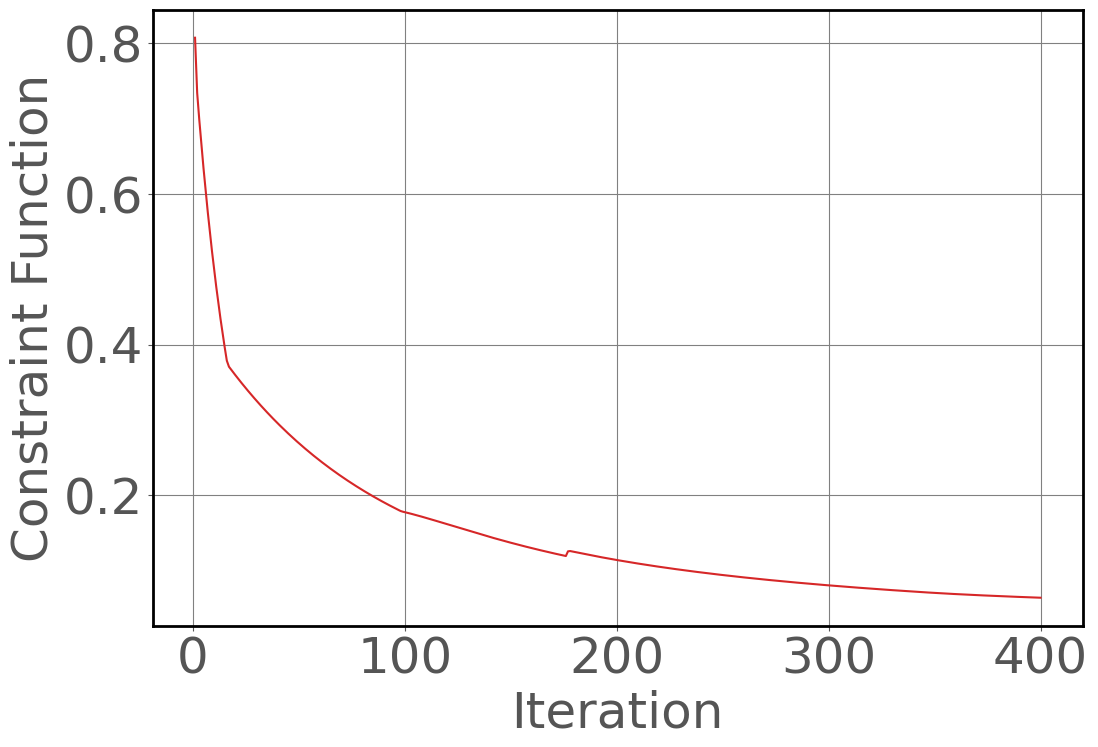}        
        \caption{Sum of constraint violation}
        \label{fig:constraint_violations}
    \end{subfigure}
    \hfill
    \begin{subfigure}[t]{0.3\textwidth}
        \includegraphics[width=\textwidth]{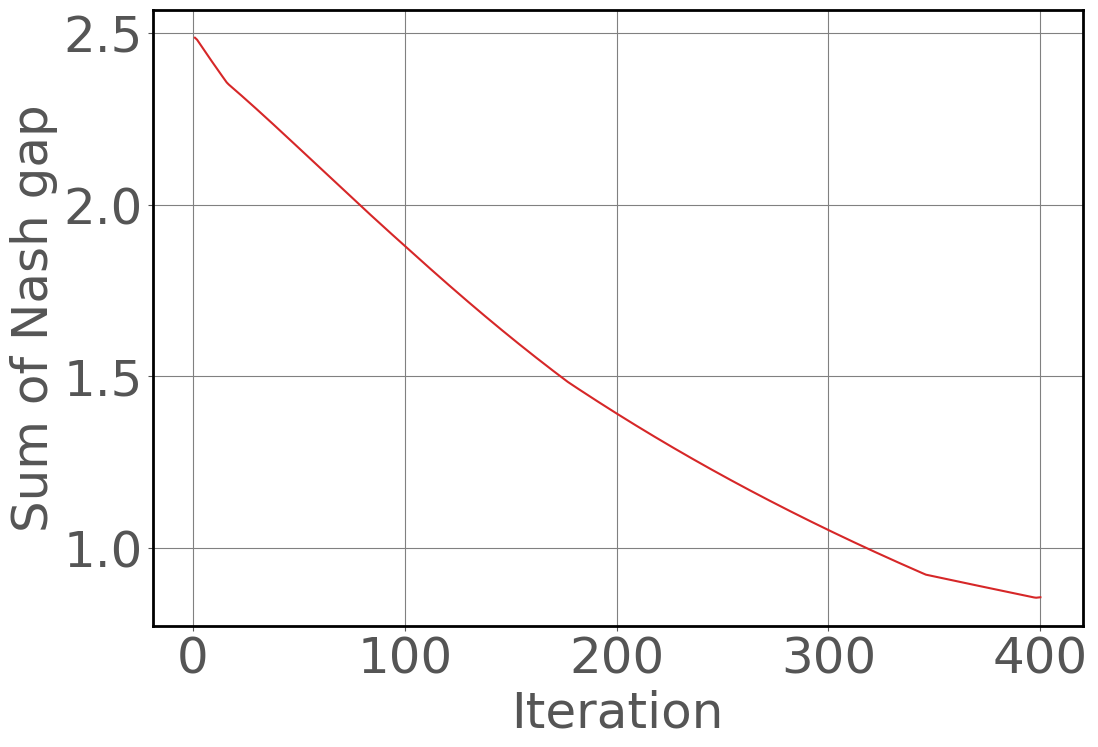}
        \caption{Nash gap}
        \label{fig:nash_gap}
    \end{subfigure}
    \caption{Collection of plots that are used to identify an approximate feasible approximate Nash Equilibrium.}
    \label{fig:plots}
\end{figure}


\paragraph{Evaluation}
To assess the solution's quality and the convergence speed of our algorithm, we employ several metrics of convergence: constraint violation, Lagrange multipliers, and Nash gap. Detailed calculations and information regarding the Nash Gap and hyperparameter selection can be found in the Appendix. In the final iteration, we anticipate the Nash gap to approach zero, constraint violation to be close to zero, and the Lagrange multipliers to have converged to a finite value. The plots are presented in \cref{fig:plots}.
\section{Conclusions}

In conclusion, this paper delves into decentralized computation of approximate Nash equilibria in constrained potential games, offering valuable insights and avenues for future research. We emphasize the importance of improving the convergence rates, based on our current findings. Furthermore, exploring different configurations of coupled constraints shows promise for advancing our understanding in this field. Additionally, extending the applicability of our results to scenarios beyond traditional potential games presents intriguing opportunities for further research.

\section*{Acknowledgments}
Ioannis Panageas would like to acknowledge startup grant from UCI. Part of this work was conducted while Nikolas, Fivos and Ioannis were visiting Archimedes Research Unit. This work has been partially supported by project MIS 5154714 of the National Recovery and Resilience Plan Greece 2.0 funded by the European Union under the NextGenerationEU Program.

\newpage
\bibliographystyle{alpha}
\bibliography{references}

\appendix
\section{Additional Definitions and Propositions} \label{section:appendix_definition}

In this section, we provide the definitions for the \textit{per-player} primal and dual problems, where the joint strategy profile of the rest of the players, denoted by $\mathbf{\hat{x}}_{-i}$, is assumed to be fixed. Additionally, we present some auxiliary lemmas that are used to establish the main theorem of this work, \cref{thm:main}.


\begin{definition}[Primal problem - per player]
\label{def:primal_per_player}
Let each player $i \in \calN$ have a set of $d_i$ private constraints, denoted by $g_{i,m}(\cdot)$ for any $m \in [d_i]$. Then, given a fixed joint strategy profile for the other players $\vxhat_{-i}$, player $i$ aims to solve the following optimization problem.

\begin{equation} \label{eq:primal_per_player}
\begin{array}{ccc}
  \text{minimize}   & \Phi( \vx_i, \vxhat_{-i}) & \\
  \text{subject to} & g_{i,m}( \vx_i ) \leq 0, & m = \{ 1, 2, \ldots, d_i 
  \}
  \\
\end{array}
\end{equation}

\end{definition}

\begin{definition}[Dual problem - per player]
\label{def:dual_per_player}
Let $i \in \calN $ be an arbitrary player associated with a primal problem, as defined in \cref{def:primal_per_player}. Then, the dual (Lagrangian) problem is defined as follows.

\begin{equation} \label{eq:dual_per_player}
\begin{array}{ccc}
  \text{maximize}   & \inf\limits_{\bm{x}_i \in \Delta(\mathcal{A}_i)} \mathcal{L}_i(\bm{x}_i, \bm{\lambda}_i; \vxhat_{-i}, \vlambdahat_{-i}) & \\
  \text{subject to} & \bm{\lambda}_{i} \succeq \bm{0}
\end{array}
\end{equation}

\noindent where $
\calL_i(\bm{x}_i, \bm{\lambda}_i; \bm{\hat{x}}_{-i},
\vlambdahat_{-i}
) 
\coloneqq 
\Phi(\bm{x}_i, \bm{\hat{x}}_{-i})
+ 
\sum\limits_{m=1}^{d_i} \lambda_{i,m} g_{i,m}(\vx_i)
$ and $\bm{\lambda}_i \in \mathbb{R}^{d_i}_{+}$.
\end{definition}

In a similar way, we can define the per-player regularized Lagrangians $\Tilde{\calL}_i$ by adding the regularization term $- \mu \norm{\vlambda_i}^2$.


\begin{proposition}[Claim C.2 \cite{leonardos2021global}]
\label{prop:spectral_norm_block_matrix}
Consider a symmetric block matrix $C$ with $n \times n$ matrices so that $\norm{C_{ij}}_2 \leq L$. Then, it holds that $\norm{C}_2 \leq n L$. In other words, if all block matrices have spectral norm at most $L$, then $C$ has spectral norm at most $nL$.
\end{proposition}

The following lemma establishes that given a point $x$ for which $\norm{C}_2$ is bounded by $\epsilon$, we can easily find a new point $x^+$ that is an approximate first-order stationary point, as defined in \cref{def:approximate_stationary_points}. The proof follows a similar approach as in Proposition B.1 in \cite{agarwal2020optimality}.

\begin{proposition} \label{prop:cone_to_nash}
Let $f$ be a $\beta$-smooth function. Define the gradient mapping 

\begin{equation}
G^{\eta}(x) = \frac{1}{\eta} 
\left(
x - P_C (x - \eta \nabla f(x))
\right)
\end{equation}

and the updated rule for the projected gradient is $x^{+} = x - \eta G^{\eta}(x)$. If $\norm{G^{\eta} (x)}_2 \leq \epsilon$, then

\begin{equation}
- \min\limits_{x + \delta \in \Delta, \norm{\delta} \leq 1} \delta^{\top} \nabla f(x) \leq \epsilon (\eta \beta + 
1)
\end{equation}
\end{proposition}

\begin{proof}
From \cref{thm:gradient_in_cone} we get

\begin{equation}
- \nabla f(x^{+}) \in N_{\Delta}(x^{+}) + \epsilon (\eta \beta + 1) B_2
\end{equation}

where $B_2$ is the unit ball $\ell_2$ ball, and $N_C$ is the normal cone of the convex set $C$. The normal cone $N_C(x^{+})$ is defined as $N_C(x^{+}) = \{ w \, | \, \langle w, y-x^{+} \rangle \leq 0, \, y \in C \}$, where $(y-x^{+})$ represents a vector in the tangent cone $T_C(x^{+}) = \textrm{cl} \big( \{ s \, (y-x^{+}) \, | \, y \in C, s \geq 0 \} \big)$. Since $- \nabla f(x^{+})$ is of $\epsilon (\eta \beta + 1)$ distance from the normal cone the normal cone, we can deduce that $\langle - \nabla f(x^{+}), y-x^{+} \rangle \leq \epsilon (\eta \beta + 1)$ for any vector $y \in C$. This inequality holds true for any vector $\delta = y - x^{+}$ in the tangent cone, including the vector that attains the maximum value.
\begin{align}
\max\limits_{x + \delta \in \Delta, \norm{\delta} \leq 1} \delta^{\top} \big( - \nabla f(x) \big) &\leq \epsilon (\eta \beta + 1) \nonumber \\
\Rightarrow \nonumber \\
- \min\limits_{x + \delta \in \Delta, \norm{\delta} \leq 1} \delta^{\top} \nabla f(x) &\leq \epsilon (\eta \beta + 1) 
\end{align}
\end{proof}

\begin{remark} \label{remark:appendix}
When the step size $\eta$ is set to $1/\beta$, we observe that the upper bound in \cref{prop:cone_to_nash} is $2\epsilon$. 
\end{remark}
\section{Omitted Proofs from \texorpdfstring{\cref{section:main_result}}{Section 4}} 
\label{section:appendix_proofs}

In this section, we provide the proofs for the statements that were presented in Section~\ref{section:analysis}. We present the proofs in the same order as they were in the main text.

\begin{lemma}[Proof of \cref{lemma:unique_maximizer}]
Let $\vxhat$ be an arbitrary joint strategy profile. Then, the function $\Tilde{\calL} (\vxhat, \vlambda) = \calL (\vxhat, \vlambda) - \mu \norm{\vlambda}^{2}$ is strongly concave, and so the maximizer $\vlambdahat = \argmax_{\vlambda} \Tilde{\calL} (\vxhat, \vlambda) $ is unique.
\end{lemma}

\begin{proof}
Since the Lagrangian $\calL (\vxhat, \vlambda)$ is linear in $\vlambda$ to prove that the regularized Lagrangian $\Tilde{\calL}$ is strongly concave, it suffices to prove that the regularization term, $\mu \norm{ \vlambda }^2$, is strongly concave. It can be easily shown that $h(\vlambda) = \mu \norm{\vlambda}^2$ is strongly concave, as the Hessian matrix $\nabla^2 h(\vlambda) = 2 \mu \bm{I}$ is positive definite.\\ \\
Regarding the uniqueness of the maximizer, it readily follows as a strongly concave function is necessarily strictly concave as well.
\end{proof}

\begin{lemma}[Proof of \cref{lemma:bounded_maximizer}]
Let $\vxhat$ be an arbitrary joint strategy profile. Then, the maximizer $\vlambdahat = \argmax_{\vlambda} \{ \Tilde{\calL}( \vxhat, \vlambda) \} $ has bounded norm, i.e. $\norm{\vlambdahat} \leq \frac{\sqrt{d} G_{\maxtag}}{2 \mu} = \Lambda_{\maxtag}$, where \\ $G_{\max} = \max\limits_{i \in \calN} \max\limits_{m \in [d_i]} \max\limits_{\vx_i \in \Delta(\calA_i)} g_{i,m}(\vx_i)$.
\end{lemma}

\begin{proof}
Let $\vlambdahat$ be the maximizer of $\Tilde{\calL}( \vxhat, \cdot)$. According to \cref{lemma:unique_maximizer}, we know that $\vlambdahat$ is unique and so the gradient of $\Tilde{\calL}( \vxhat, \cdot)$ must vanish at $\vlambdahat$.
\begin{align}
\nabla_{\vlambda} 
\left(
\Tilde{\calL} ( \vxhat, \vlambdahat)
\right)
&=
\nabla_{\bm{\lambda}} 
\bigg(
\Phi(\vxhat) 
+
\vlambdahat^{\top} \vg(\vxhat)
\notag
\\
&\phantom{=\nabla_{\bm{\lambda}}} 
- 
\mu \norm{ \vlambdahat }^2
\bigg)
\\
\Rightarrow
\vg(\vxhat)
- 
2 \mu \vlambdahat
&=
\bm{0}
\\
\Rightarrow
\norm{\vlambdahat} 
&=
\frac{\norm{\vg(\vxhat)}}{2 \mu}
\end{align}

Then, for any player $i \in [n]$ and any constraint $m \in \{ 1, 2, \ldots, d_i \}$ it holds that
$
g_{i,m} (\vx_i)
\leq
G_{\maxtag}^i
$ where $G_{\maxtag}^i = \max\limits_{\bm{x}_i \in \Delta(\mathcal{A}_i)} \max\limits_{m \in [d_i]} g_{i,m}(\vx_i)$ is the maximum value over all $\bm{x}_i \in \Delta(\mathcal{A}_i)$ and over all constraints of player $i$. Hence, defining $G_{\maxtag} = \max_{i \in \calN} G_{\maxtag}^i$ as the maximum over all players, and $\Lambda_{\maxtag} = \frac{G_{\maxtag}}{2 \mu}$ we can upper bound the norm of $\vlambdahat$. 
\begin{align}
\norm{\vlambdahat} &= \frac{1}{2 \mu} \sqrt{ \sum\limits_{i \in \calN} \sum\limits_{m \in [d_i]} \left( g_{i,m} (\vx_i) \right)^2 } \nonumber \\
&\leq \frac{1}{2 \mu} \sqrt{ \sum\limits_{i \in \calN} \sum\limits_{m \in [d_i]} \left( G_{\maxtag} \right)^2 } \nonumber \\
&= \left( \frac{\sqrt{d} \, G_{\maxtag}}{2\mu} \right)
\end{align}
\end{proof}

\paragraph{Smoothness Properties} The following three lemmas deal with the smoothness property of $\Phi$, $\mathcal{L}$, and $\Tilde{\mathcal{L}}$. It is worth noting that this property plays a pivotal role in establishing the smoothness of $\phi(\vx) = \max_{\vlambda} \Tilde{\mathcal{L}}(\vx, \vlambda)$, thereby enabling us to utilize the \cref{lemma:descent_lemma} in subsequent steps.

\begin{lemma}[Proof of \cref{lemma:smoothness_potential}]
The potential function $\Phi$ is $(n A_{\maxtag} \Phi_{\maxtag})$-smooth. 
\end{lemma}
\begin{proof}
To establish the $\beta$-smoothness of the function $\Phi$, it is necessary to demonstrate that the eigenvalues of its Hessian matrix are smaller than $\beta$, assuming that $\Phi$ is twice differentiable. Considering the symmetry of the Hessian matrix, we can alternatively prove this by utilizing the spectral norm. Specifically, the spectral norm of a symmetric matrix is equivalent to the maximum absolute value of its largest eigenvalue. Therefore, if the Hessian matrix of $\Phi$ has a bounded spectral norm, it implies that all eigenvalues of the Hessian are bounded by $\beta$. Hence, we consider the Hessian $\nabla^{2} \Phi$ as a block matrix of $n \times n$ submatrices of the form $C_{ij} = \nabla^{2}_{\vx_i, \vx_j} \Phi$. 

\begin{align}
\norm{\nabla^{2}_{\vx_i, \vx_j} \Phi}_2
&\leq
\norm{\nabla^{2}_{\vx_i, \vx_j} \Phi}_{F}
\notag
\\
&=
\left(
\sum\limits_{k=1}^{|\calA_i|}
\sum\limits_{\ell=1}^{|\calA_j|}
\bigg|
\frac{\partial^2 \Phi}{\partial x_{ik} \, \partial x_{i \ell}} 
\bigg|^{2}
\right)^{1/2}
\label{aux:smoothness_potential_1}
\\
&\leq
\left(
\sum\limits_{k=1}^{|\calA_i|}
\sum\limits_{\ell=1}^{|\calA_j|}
| \Phi_{\maxtag} |^2
\right)^{1/2}
\notag
\\
&=
\left(
|\calA_i|
|\calA_j|
| \Phi_{\maxtag} |^2
\right)^{1/2}
\label{aux:smoothness_potential_2}
\\
&\leq
A_{\maxtag} \Phi_{\maxtag}
\label{aux:smoothness_potential_3}
\end{align}

\cref{aux:smoothness_potential_1} follows from a well-known inequality $\norm{A}_2 = \sigma_{\maxtag} (A) \leq \norm{A}_{F}$, and \cref{aux:smoothness_potential_3} holds from the definition of $A_{\maxtag}$. Regarding \cref{aux:smoothness_potential_2}, we notice that $\frac{\partial^2 \Phi}{\partial x_{ik} , \partial x_{i \ell}}$ is bounded by $\Phi_{\maxtag}$, which represents the maximum value of the potential. This can be understood as the case in which players $i$ and $j$ assign all their probability mass to strategies $k$ and $\ell$ respectively. 
Using now \cref{prop:spectral_norm_block_matrix}, we conclude that 
\[
\norm{\nabla^2 \Phi}_2 \leq n \norm{\nabla^2_{\vx_i \vx_j} \Phi}_2 \leq n A_{\maxtag} \Phi_{\maxtag}
\]
\end{proof}

\begin{lemma}[Proof of \cref{lemma:smoothness_lagrangian}]
The Lagrangian function $\calL$ is $(n A_{\maxtag} \Phi_{\maxtag} + \Lambda_{\maxtag} \gamma)$-smooth,
where $\gamma$ is the smoothness of the constraints.
\end{lemma}

\begin{proof}
By combining the  triangle inequality, Cauchy–Schwarz inequality and the fact that $\norm{\nabla^2 \Phi} \leq n A_{\maxtag} \Phi_{\maxtag}$ and $\norm{\vlambda \nabla^2 \vg(\vx)} \leq \Lambda_{\maxtag} \gamma$, we can bound to spectral norm of sum of those matrices.
\end{proof}

\begin{lemma}[Proof of \cref{lemma:smoothness_reg_lagrangian}]
The regularized Lagrangian $\Tilde{\calL}$ is $(n A_{\maxtag} \Phi_{\maxtag} + \Lambda_{\maxtag} \gamma + 2 \mu)$-smooth.
\end{lemma}

\begin{proof}
As in the proof of \cref{lemma:smoothness_lagrangian}, we simply apply the triangle-inequality in the spectral norm and get that $\Tilde{\calL}$ is $(n A_{\maxtag} \Phi_{\maxtag} + \Lambda_{\maxtag} \gamma + 2 \mu)$-smooth.
\end{proof}

\paragraph{Boundness of optimal Lagrange multipliers} The significance of the next lemma lies in its ability to provide a bound value of the Lagrange multipliers $\vlambda^{\star}(\vxhat) = \vlambdahat$, as defined in \cref{lemma:smoothness_max}, in the case of an approximate stationary point $\vxhat$. However, the original bound we presented in \cref{lemma:bounded_maximizer} is ineffective in this scenario. This is because the bound is dependent on the regularization parameter, which, in turn, relies on the desired precision of the approximation we aim to achieve.

\begin{lemma}[Proof of \cref{lemma:bounded_optimal_multipliers}]
Let $\vxhat$ represent the $\epsilon$-approximate first-order stationary point returned by the Algorithm. Then, for each set of Lagrange multipliers $\vlambdahat_i$, where $\vlambdahat_i = \argmax_{\vlambda_i} \calL_i(\vxhat_i, \vlambda_i ; \vxhat_{-i}, \vlambdahat_{-i})$, we can bound them by $\frac{2 \big( \Phi_{\maxtag} - \Phi_{\mintag} \big) }{\xi_{i,m}}$ component-wise, where $\xi_{i,m}$ are defined in \cref{assumption:slater}.
\end{lemma}

\begin{proof}
The proof uses strong duality, which is stated in \cref{lemma:strong_duality_original}. 
\begin{align}
\Phi_{\mintag} 
&\leq
\min\limits_{\vx_i \in \calS_i} \Phi(\vx_i, \vxhat_{-i})
\label{aux:bounded_multiplier_1}
\\
&=
\max\limits_{\vlambda_i} \min\limits_{\vx_i} 
\calL (\vx_i, \vlambda_i ; \vxhat_i, \vlambdahat_i )
\label{aux:bounded_multiplier_2}
\\
&=
\calL (\vx_i^{\star}, \vlambda_i^{\star} ; \vxhat_i, \vlambdahat_i )
\label{aux:bounded_multiplier_3}
\\
&\leq
\calL (\vxhat_i, \vlambdahat_i ; \vxhat_i, \vlambdahat_i ) + \epsilon
\label{aux:bounded_multiplier_4}
\\
&\leq
\calL (\vx_i, \vlambdahat_i ; \vxhat_i, \vlambdahat_i ) + \epsilon & \text{for any } \vx_i \in \Delta(\calA_i)
\label{aux:bounded_multiplier_5}
\end{align}
The validity of \cref{aux:bounded_multiplier_1} can be attributed to the fact that the right-hand side is defined exclusively within the feasibility set of player $i$. As for \cref{aux:bounded_multiplier_2}, it can be established through strong duality, which states that the optimal value of the dual problem \eqref{eq:dual_problem} is equal to the optimal value of the primal problem \eqref{eq:primal_problem}. Regarding \cref{aux:bounded_multiplier_3}, both variables $\vx$ and $\vlambda$ are confined within bounded domains, ensuring that the optimal values remain finite. As for \cref{aux:bounded_multiplier_4}, it holds true since $(\vxhat_i, \vlambdahat_i)$ is an $\epsilon$-approximate Nash equilibrium, meaning that $\calL (\vxhat_i, \vlambdahat_i ; \vxhat_i, \vlambdahat_i )$ differs from the minmax (or maxmin) value by at most $\epsilon$. Lastly, \cref{aux:bounded_multiplier_5} follows from the definition of an approximate Nash equilibrium. Specifically, we can interpret $\calL (\vx_i, \vlambda_i ; \vxhat_i, \vlambdahat_i )$ as a zero-sum game between the minimizer $\vx_i$ and the maximizer $\vlambda_i$. Thus, according to \cref{def:approximate_nash_eq}, it can be deduced that no unilateral deviations $\vx_i$ can yield improvements greater than $O(\epsilon)$.
\begin{equation}
\calL (\vxhat_i, \vlambdahat_i ; \vxhat_i, \vlambdahat_i ) \leq 
\calL (\vx_i, \vlambdahat_i ; \vxhat_i, \vlambdahat_i )
+ \epsilon
\quad
\text{for any }
\vx_i \in \Delta(\calA_i)
\end{equation}
To proceed, we leverage Slater's condition as stated in \cref{assumption:slater}. In particular, since \cref{aux:bounded_multiplier_5} holds for any $\vx_i$, we are free to choose $\vx_i = \Tilde{\vx}_i$.
\begin{align}
\Phi_{\mintag} 
&\leq
\calL (\Tilde{\vx}_i, \vlambdahat_i ; \vxhat_i, \vlambdahat_i ) + \epsilon
\\
&=
\Phi (\Tilde{\vx}_i, \vxhat_i) 
+
\vlambda_i^{\top} \vg_i(\Tilde{\vx}_i) + 2 \epsilon
\\
&\leq
\Phi_{\maxtag} + \sum\limits_{m=1}^{d_i} \lambda_{i,m} g_{i,m}(\Tilde{\vx}_i) + 2 \epsilon
\\
&\leq
\Phi_{\maxtag} + \sum\limits_{m=1}^{d_i} \lambda_{i,m} \big( - \xi_{i,m} \big) + 2 \epsilon 
\\
&\phantom{\leq}
\text{(Slater's condition)}
\\
\Rightarrow
\sum\limits_{m=1}^{d_i} \lambda_{i,m} \xi_{i,m} 
&\leq
\Phi_{\maxtag} - \Phi_{\mintag} + \epsilon
\\
\Rightarrow
\lambda_{i,m} 
&\leq 
\frac{\Phi_{\maxtag} - \Phi_{\mintag}}{\xi_{i,m}} 
+ \epsilon
\quad 
\text{ for any } m \in [d_i]
\\
\Rightarrow
\lambda_{i,m} 
&\leq 
\frac{2 \big( \Phi_{\maxtag} - \Phi_{\mintag} \big) }{\xi_{i,m}} 
\quad 
\text{ for any } m \in [d_i]
\end{align}

where the last inequality holds for sufficiently small $\epsilon$.
\end{proof}

\subsection{Full Proof of Lemma~\ref{thm:stationary_reg_lagrangian}}
Here, we present an expanded proof of the main \cref{thm:stationary_reg_lagrangian} and address certain points that were initially omitted in the main text due to space limitations.

\begin{lemma}[Full Proof of \cref{thm:stationary_reg_lagrangian}]
Let the learning rate $\eta$ be $\frac{1}{\beta}$, where
$$
\beta 
=
\frac{1}{c}
\quad
\text{and}
\quad 
c = 4((n A_{\maxtag})^2 + (\Lambda_{\maxtag} \gamma)^2)
$$
If we run Algorithm \igd for \\ $T = \frac{32}{\epsilon^2 \mu}
\big( \Phi_{\maxtag} + \Lambda_{\maxtag} \sqrt{d} G_{\maxtag} \big)
\left(
(n A_{\maxtag})^2 + (\Lambda_{\maxtag} \gamma)^2
\right)$
iterations,
then there exists a timestep $t \in \{ 1, 2, \cdots, T \}$ such that $\vx^{(t)}$ is an $\epsilon$-approximate first order stationary point of $\phi$.
\end{lemma} 

\begin{proof}
The first step involves demonstrating that the function $\phi(\cdot) = \max_{\vlambda} \Tilde{\calL} (\cdot, \vlambda)$ is a smooth function. Based on the result from \cref{lemma:smoothness_reg_lagrangian}, we know that $\Tilde{\calL}$ is $\ell$-smooth, with $\ell = (n A_{\maxtag} \Phi_{\maxtag} + \Lambda_{\maxtag} \gamma + 2 \mu)$. Utilizing the smoothness result in \cref{lemma:smoothness_max}, we can conclude that $\phi(\cdot)$ is also smooth, with a parameter $(\ell + \kappa \ell)$, where $\kappa = \ell / \mu$ is the conditional number. Therefore, in order to apply the descent lemma (\cref{lemma:descent_lemma}), it is necessary to establish an upper bound for the smoothness parameter.

\begin{align}
(\ell + \kappa \ell) 
= 
\left(
\ell + \frac{\ell^2}{\mu}
\right)
&\leq 
2 \frac{\ell^2}{\mu} 
\notag
\\
&=
\frac{2}{\mu} 
\left( 
(n A_{\maxtag}) + (\Lambda_{\maxtag} \gamma) 
\right)^2
\label{eq:lemma_app1}
\\
&\leq \frac{4 \left( (n A_{\maxtag})^2 + (\Lambda_{\maxtag} \gamma)^2 \right)}{\mu} 
\label{eq:lemma_app2}
\\
&= 
\frac{4 c}{\mu}
=
\beta
\end{align}
where inequality \ref{eq:lemma_app1} holds for a sufficiently small choice of $\mu$. Additionally, inequality \ref{eq:lemma_app2} follows from the well-known inequality $(a + b)^2 \leq 2(a^2 + b^2)$. Therefore, for the chosen value of $\beta = \frac{4c}{\mu}$, we can readily establish an upper bound for the smoothness parameter.\\ \\
We can now utilize the descent lemma on $\phi$, which establishes that the sequence of successive steps will form a non-increasing sequence, as defined in \cref{item:descent_gradient_1} in \cref{thm:descent_gradient}. This can be inferred from \cref{lemma:descent_lemma}, which ensures that $\phi(\vx^{(t+1)}) - \phi(\vx^{(t)}) \leq - \frac{1}{2\beta} \norm{\vx^{(t+1)} - \vx^{(t)}}_{2}^{2} \leq 0$. Moreover, \cref{item:danskin_3} in \cref{thm:descent_gradient} guarantees the existence of at least one timestep $t \in \{ 1, 2, \cdots, T \}$ where $\norm{\vx^{(t+1)} - \vx^{(t)}}_{2}^{2} \leq \frac{2 \beta \delta_{\phi}}{T}$, with $\delta_{\phi} = \phi(\vx^{(0)}) - \phi(\vx^{\star})$. \\ \\
To proceed, we have to bound the difference $\delta_{\phi}$. Since, we cannot assume that the initial point is feasible, the easiest bound $\delta_{\phi}$ can be derived as follows:
\begin{align}
\delta_{\phi} &\leq \max_{\vx \in \Delta} \phi(\vx) \nonumber \\
&\leq \Phi_{\maxtag} + \vlambda^{\top} \vg(\vx) \nonumber \\
&\leq \Phi_{\maxtag} + \norm{\vlambda} \norm{\vg(\vx)} \nonumber \\
&\leq \Phi_{\maxtag} + \Lambda_{\maxtag} \sqrt{d} G_{\maxtag}
\end{align}
where the second-to-last inequality follows from the Cauchy-Schwarz inequality, while the last inequality is derived from \cref{lemma:bounded_maximizer}. Hence, if we bound $\norm{\vx^{(t+1)} - \vx^{(t)}}_{2}^{2}$ by $\epsilon^2 / 4$, it follows
\begin{align}
\norm{\vx^{(t+1)} - \vx^{(t)}}_{2}^{2} 
&\leq 
\frac{2 \beta}{T} \big( \Phi_{\maxtag} + \Lambda_{\maxtag} \sqrt{d} G_{\maxtag} \big)
\notag
\\
&\leq 
\epsilon^2/4 
\\
\Rightarrow 
T &\geq \frac{8 \beta}{\epsilon^2} 
\big( \Phi_{\maxtag} + \Lambda_{\maxtag} \sqrt{d} G_{\maxtag} \big)
\notag
\\
&=
\frac{32}{\epsilon^2 \mu}
\big( \Phi_{\maxtag} + \Lambda_{\maxtag} \sqrt{d} G_{\maxtag} \big)
\notag
\\
&\phantom{=} \times
\left(
(n A_{\maxtag})^2 + (\Lambda_{\maxtag} \gamma)^2
\right)
\end{align}
we get that $\norm{\vx^{(t+1)} - \vx^{(t)}}_{2} \leq \epsilon/2$. Then, from \cref{prop:cone_to_nash} and \cref{remark:appendix} we conclude that $\vx^{(t+1)}$ is an $\epsilon$ approximate first order stationary point of $\phi$.

\begin{equation}
- \min\limits_{
(\vx + \vdelta) \in \Delta, 
\norm{\vdelta}^{2} \leq 1
}
\vdelta^{\top} \nabla_{\vx} \phi(\vx^{(t+1)}) 
\leq 
\epsilon
\end{equation}

\end{proof}

\subsection{Approximate Feasibility Approximate Stationarity}
In this subsection, we present the final part of the proof. As we pointed out in the main text, \cref{lemma:smoothness_max}, which is derived from an application of \cref{thm:danskin_theorem}, states that $\nabla \phi(\cdot) = \nabla_{\vx} \Tilde{\calL} (\cdot, \vlambda^{\star}(\cdot))$, where $\vlambda^{\star}(\cdot)$ represents the unique maximizer for a given $\vx$, as indicated in Lemma \ref{lemma:unique_maximizer}. This in turn implies that the point $\vxhat$ obtained from the algorithm is an approximate first order stationary point of the regularized Lagrangian. \\ \\
Utilizing \cref{lemma:bounded_optimal_multipliers} we can also show that for the case of an approximate  first order stationary point $\vxhat$, the Lagrange multipliers $\vlambda^{\star}(\vxhat) = \vlambdahat$ of $\Tilde{\calL}$ can be bounded by a term that is independent of $\mu = O(\epsilon)$. This result guarantees that $\vlambdahat$ must reside in the interior of its domain.\\ \\
Finally, we need to argue about the optimality of $(\vxhat, \vlambdahat)$ with respect to $\Tilde{\calL}$ and as a result with respect to $\Phi$ and then $C_i$, i.e. the cost function of each player.

\begin{lemma}[Proof of \cref{lemma:approximate_optimality}]
Given an $\epsilon$ approximate stationary point $(\vxhat, \vlambdahat)$ of $\Tilde{\calL}$, we get that $\vxhat$ is $O(\epsilon)$ approximate feasible approximate Nash equilibrium, as in defined \cref{def:approximate_feasible_approximate_nash}.
\end{lemma}

\begin{proof}
Consider $(\vxhat, \vlambdahat)$ as an $\epsilon$-approximate stationary point of $\Tilde{\mathcal{L}}$. According to \cref{lemma:unique_maximizer}, it follows that the gradient of $\Tilde{\mathcal{L}}$ with respect to $\vlambda$ must vanish, leading to the first important bound.
\begin{align}
\nabla_{\vlambda} 
\left(
\Phi(\vxhat) + \vlambdahat^{\top} \vg(\vxhat) - \mu \norm{\vlambdahat}^2
\right)
&= 
\bm{0}
\\
\Rightarrow
\nabla_{\vlambda_i} 
\left(
\Phi(\vxhat) + \vlambdahat^{\top} \vg(\vx) - \mu \norm{\vlambdahat}^2
\right)
&=
\bm{0}
&
\text{ for any $i \in \calN$}
\label{eq:main_thm_1}
\end{align}
Let $i \in \calN$ denote an agent. By utilizing \cref{eq:main_thm_1} and considering the fact that $\vlambdahat^\top \vg(\vxhat) = \sum_{i \in \calN} \vlambdahat_i^\top \vg_i(\vxhat_i)$, we can conclude:
\begin{equation}
\vg_i(\vxhat_i) - 2 \mu \vlambdahat_i = 0
\Rightarrow
\norm{\vg_i(\vxhat_i)} \leq 2 \mu \norm{\vlambdahat_i}
\end{equation}
We note here that \cref{lemma:bounded_optimal_multipliers} refers to the Lagrange multiplier with respect to the original/non regularized Lagrangian function $\calL$, meaning $\vlambda' = \max_{\vlambda} \calL(\vxhat, \vlambda)$, and \textit{not} with respect to the regularized version $\Tilde{\calL}$. However, the regularization term in  $\Tilde{\calL}$ can only restrict the value of the multiplier and so it guarantees that the norm of the optimal ones $\vlambda^{\star}(\vxhat) = \vlambdahat$ w.r.t to $\Tilde{\calL}$ is necessarily bounded by the norm of the $\vlambda'$.
As a result, we conclude that $\norm{\vg_i(\vxhat_i)}$ is bounded by $O(\mu) = O(\epsilon)$.  This implies that $\vxhat_i$ is an approximate feasible point.\\ \\
To proceed, we need to show that $\vxhat$ constitutes an approximate Nash equilibrium, as it is defined in \cref{def:approximate_nash_eq}. Again, we consider the per-player regularized Lagrangian $\Tilde{\calL}_i$.
\begin{align}
\Tilde{\calL}_i 
(\vxhat_i, \vlambdahat_i ; \vxhat_{-i}, \vlambdahat_{-i})
&\leq 
\min\limits_{\vx_i \in \Delta(\calA_i)}
\max\limits_{\vlambda}
\Tilde{\calL}_i
(\vx_i, \vlambda_i ; \vxhat_{-i}, \vlambdahat_{-i})
\notag
\\
&\phantom{\leq} \text{where } \norm{\vlambda} \leq \Lambda_{\maxtag}
+ \epsilon
\label{eq:main_lemma_app_1}
\\
&\leq
\min\limits_{\vx_i \in \Delta(\calA_i)}
\Phi(\vx_i, \vxhat_i) + \epsilon
\label{eq:main_lemma_app_2}
\end{align}
The similar argument we presented in the proof of \cref{lemma:bounded_maximizer} applies here. Specifically, since $(\vxhat_i, \vlambdahat_i)$ is an approximate stationary point, the value of $\Tilde{\calL}_i$ at that point cannot differ more than $\epsilon$ to the $\min\max$ value, and so \cref{eq:lemma_app1} holds. Regarding \cref{eq:lemma_app2}, we use strong duality and specifically \cref{eq:strong_duality} shown in \cref{remark:1}. Now, as $\{ \vx_i \in \Delta(\mathcal{A}_i) \}$ consists all the approximate $O(\epsilon)$ feasible points $\{ \vx_i \, | \, \vg_i(\vx_i) \leq O(\epsilon)\}$, the minimum value of $\Phi$ over the latter set is at least as large as the minimum value of $\Phi$ over the former set.
\footnote{In other words, we claim that if $X \subseteq Y$ then $\min (X) \geq \min (Y)$.}
\begin{equation}
\Tilde{\calL}_i 
(\vxhat_i, \vlambdahat_i ; \vxhat_{-i}, \vlambdahat_{-i})
\leq
\min\limits_{ \{  \vx_i \, | \, \vg_i(\vx_i) \leq O(\epsilon) \} }
\Phi(\vx_i, \vxhat_i)
+ \epsilon
\end{equation}
For bounding the $\Tilde{\calL}_i$ in the opposite direction, we follow a different way; we utilize \cref{lemma:bounded_optimal_multipliers} so that we can bound $\vlambdahat_i$ with a term that is independent of $\mu = O(\epsilon)$.
\begin{align}
\Tilde{\calL}_i 
(\vxhat_i, \vlambdahat_i ; \vxhat_{-i}, \vlambdahat_{-i})
&=
\Phi(\vxhat_i, \vxhat_{-i})
+
\vlambdahat_i^{\top} \vg_i(\vxhat_i)
- \mu \norm{\vlambdahat_i}
\\
&\geq 
\Phi(\vxhat_i, \vxhat_{-i})
- \mu \norm{\vlambdahat_i}
\label{eq:lemma_app3} 
\\
&\geq 
\Phi(\vxhat_i, \vxhat_{-i})
- O(\epsilon)
\label{eq:lemma_app4}\\
&\leq
\min\limits_{ \{  \vx_i \, | \, \vg_i(\vx_i) \leq O(\epsilon) \} }
\Phi(\vx_i, \vxhat_i)
-
O(\epsilon)
\end{align}
\cref{eq:lemma_app3} holds because $\vlambdahat_i^{\top} \vg_i(\vxhat_i)$ is always non-negative. This is true because $\vlambdahat_i \succeq \bm{0}$ and it is only zero only if $\vxhat_i$ is feasible, meaning $\vg_i(\vxhat_i) \leq 0$. \cref{eq:lemma_app4} follows from \cref{lemma:bounded_optimal_multipliers} and also because $\mu = O(\epsilon)$. Therefore, we have demonstrated that for any unilateral $O(\epsilon)$ approximate feasible deviation, the value of the regularized Lagrangian $\Tilde{\mathcal{L}}_i$ at $(\vxhat_i, \vlambdahat_i)$ remains within $O(\epsilon)$ of the minimum value of $\Phi$ over the set of $O(\epsilon)$ approximate feasible points. The next step is to examine the value of $\Phi$ at $(\vxhat_i, \vlambdahat_i)$, which can be done straightforwardly as follows:
\begin{align}
\Tilde{\calL}_i
(\vxhat_i, \vlambdahat_i ; \vxhat_{-i}, \vlambdahat_{-i})
-
\Phi(\vxhat_i, \vxhat_{-i})
&=
\vlambdahat_i^{\top} \vg_i(\vxhat_i)
- \mu \norm{\vlambdahat_i}
\\
\quad
\big|
\Tilde{\calL}_i
(\vxhat_i, \vlambdahat_i ; \vxhat_{-i}, \vlambdahat_{-i})
-
\Phi(\vxhat_i, \vxhat_{-i})
\big|
&= 
\big|
\vlambdahat_i^{\top} \vg_i(\vxhat_i)
- \mu \norm{\vlambdahat_i}
\big|
\\
&\leq
\big|
\vlambdahat_i^{\top} \vg_i(\vxhat_i)
\big|
+
\big|
\mu \norm{\vlambdahat_i}
\big|
\\
&
\leq O(\epsilon)
\end{align}
Finally, we conclude that $\Phi(\vxhat_i, \vxhat_{-i})$ remains withing $O(\epsilon)$ of the minimum value of $\Phi$ over the set of $O(\epsilon)$ approximate feasible points, and so $\vxhat$ is an $O(\epsilon)$ approximate feasible approximate Nash equilibrium.
 
\end{proof}

\begin{theorem}[Proof of \cref{thm:main}] 
Assuming that all agents perform Algorithm \igd, after $T = O(1/\epsilon^6)$ steps, there exists an iterate $\vx^{(t)}$ for $t \in [T]$, so that $\vx^{(t)}$ is an $O(\epsilon)$ approximate feasible approximate Nash equilibrium.
\end{theorem}

\begin{proof}
The proof of \cref{thm:main} involves a straightforward application of the previously established statements. More precisely, \cref{thm:stationary_reg_lagrangian} ensures that if all agents follow Algorithm \igd, after $T = O(1/\epsilon^6)$ iterations, there exists an iterate $\vx^{(t)}$ for $t \in [T]$ that serves as an $\epsilon$ approximate first-order stationary point $\vxhat$ of $\phi(\cdot)$. Then, utilizing \cref{lemma:bounded_optimal_multipliers}, we can extend this to a pair $(\vxhat, \vlambdahat)$ that forms an $\epsilon$ approximate stationary point of $\Tilde{\mathcal{L}}$. Finally, by virtue of \cref{lemma:approximate_optimality}, it follows that $\vxhat$ constitutes an $O(\epsilon)$ approximate feasible approximate Nash equilibrium of $\Phi$.
\end{proof}
\section{Standard Optimization Background}
\label{section:appendix_optimization}

In this section we present some standard results in the literature of nonconvex smooth optimization. More specifically, we consider the following problem, where $C$ is a nonempty closed convex set.

\begin{equation} 
\begin{array}{ccc}
  \text{minimize}   & f(x) & \\
  \text{subject to} & x \in C \\ 
\end{array}
\label{eq:basic}
\end{equation}

\begin{assumption} \label{assumption:basic}
Let $f : \mathbb{R}^d \rightarrow (- \infty, \infty)$ is proper closed, $dom(f)$ is convex and $f$ is $\beta$ smooth over $\int(\dom(f))$.
\end{assumption}

\begin{lemma}[Descent Lemma] \label{lemma:descent_lemma}
Let $f$ be $\beta$-smooth function with convex domain $\mathcal{X}$. Let $x \in \mathcal{X}$, $x^{+} = P_{\mathcal{X}} \left( x - \frac{1}{\beta} \nabla f(x) \right)$
and $g_{\mathcal{X}} = \beta (x - x^{+})$. Then the following holds true:
\begin{equation}
f(x^{+}) - f(x) \leq - \frac{1}{2 \beta} \norm{g_{\mathcal{X} (x)}}^{2}_{2}.
\end{equation} 
\end{lemma}

\begin{definition}
\label{def:gradient_mapping}
We define the gradient mapping $G^{\eta}(x)$ as 
\begin{equation}
G^{\eta}(x) = \frac{1}{\eta} 
\left(
x - P_C (x - \eta \nabla f(x))
\right)
\end{equation}

where $P_C$ is the projection onto $C$. We also define the update rule for the projected gradient is $\vx^{+} = \vx - \eta G^{\eta}(\vx)$.
\end{definition}

\begin{theorem}[Theorem 10.15 \cite{beck2017first}] \label{thm:descent_gradient}
Suppose that \cref{assumption:basic} holds and let $\{ x_t \}_{t \geq 0}$ be the sequence generated by the gradient descent algorithm for solving the problem \eqref{eq:basic} with stepsize $\eta = 1/\beta$. Then,

\begin{enumerate}[label=\arabic*.]
\item \label{item:descent_gradient_1} 
The sequence $\{ f(x_t) \}_{t \geq 0}$ is non-increasing.
\item \label{item:descent_gradient_2}
$G^{\eta}(x_t) \rightarrow 0$ as $t \rightarrow \infty$.
\item \label{item:descent_gradient_3}
$\min\limits_{t=0,1,\cdots,T-1} \norm{G^{\eta}(x_t)} \leq \frac{\sqrt{2 \beta \left( f(x_0) - f(x^{\star}) \right)}}{\sqrt{T}}$
\end{enumerate}
\end{theorem}

\begin{theorem}[Lemma 3 \cite{ghadimi2016accelerated}]
\label{thm:gradient_in_cone}
Suppose that \cref{assumption:basic} holds. Let $x^{+} = x - \eta G^{\eta}(x)$ and $\norm{G^{\eta}(x)}_2 \leq \epsilon$. Then,

\begin{equation}
- \nabla f(x^{+}) \in N_C(x^{+}) + \epsilon (\eta \beta + 1) B_2
\end{equation}

where $B_2$ is the unit ball $\ell_2$ ball, and $N_C$ is the normal cone of the set $C$.
\end{theorem}

\begin{theorem}[Danskin's Theorem \cite{bertsekas1997nonlinear}] 
\label{thm:danskin_theorem}
Let $Z(x)$ be a compact subset of $\mathbb{R}^m$, and let $\phi(x) : \mathbb{R}^n \times Z \rightarrow 
\mathbb{R}$ be continuous and such that $\phi(\cdot, z) : \mathbb{R}^n \rightarrow \mathbb{R}^n$ is convex for each $z \in Z$. Then, the following statements are true.

\begin{enumerate}[label=\roman*.]
\item 
\label{item:danskin_1}
The function $f(x) = \max\limits_{z \in Z} \phi(x,z)$ is convex.

\item 
\label{item:danskin_2}
The directional derivative of $f(x)$ is given by $f'(x;y) = \max\limits_{z \in Z(x)} \phi'(x,z;y)$, where $\phi'(x,z;y)$ is the directional derivative at in $x$ in the direction of $y$, and $Z(x)$ is the set of maximizing points.

\begin{equation} \label{eq:danskin_aux}
Z(x) = 
\Bigl\{ 
\bar{z} \, | \, 
f(x,\bar{z}) = \max\limits_{z \in Z} \phi(x,z) 
\Bigl\}
\end{equation}

\item
\label{item:danskin_3}
If $Z(x)$ in \cref{eq:danskin_aux} consists of a unique point $\bar{z}$ and $\phi(x, \cdot) $ is differentiable at $x$, then $f$ is differentiable at $x$, and it holds,

\begin{equation}
\nabla_x f(x) = \nabla_x \phi(x, \bar{z}).
\end{equation}
\end{enumerate}
\end{theorem}

\end{document}